\documentclass[11pt]{amsart}
\usepackage[utf8]{inputenc}
\usepackage[T1]{fontenc}

\usepackage{amssymb}
\usepackage{amsmath}
\usepackage{color}
\usepackage{arydshln}
\usepackage{tikz-cd}
\usepackage[widespace]{fourier}
\usepackage{microtype}
\usepackage{comment}
\usepackage[colorlinks,linkcolor=blue,citecolor=teal]{hyperref}
\usepackage[capitalize]{cleveref}
\usepackage{cite}
\usepackage{booktabs}
\usepackage{todonotes}

\newtheorem{theorem}{Theorem}[section]

\newtheorem{lemma}[theorem]{Lemma}

\theoremstyle{definition}

\newtheorem{example}{Example}

\theoremstyle{remark}
\newtheorem{remark}[theorem]{Remark}

\newcommand{\K}{\mathbb{K}}

\newcommand{\PP}{\mathbb{P}}

\DeclareMathOperator{\SL}{SL}
\DeclareMathOperator{\GL}{GL}
\DeclareMathOperator{\PGL}{PGL}

\DeclareMathOperator{\Sp}{Sp}

\DeclareMathOperator{\Gr}{Gr}

\DeclareMathOperator{\stab}{stab}

\DeclareMathOperator{\Pf}{Pf}
\DeclareMathOperator{\ud}{d}

\usepackage{graphicx}
\allowdisplaybreaks[3]

\author[G. Gubbiotti]{Giorgio Gubbiotti}
\author[B. van Geemen]{Bert van Geemen}
\author[P. Vergallo]{Pierandrea Vergallo}
\address[G. Gubbiotti and P. Vergallo]{Dipartimento di Matematica ``Federigo Enriques'',
	Universit\`a degli Studi di Milano, Via C. Saldini 50, 20133
	Milano, Italy \& INFN Sezione di Milano, Via G. Celoria 16,
	20133 Milano, Italy}
\email{giorgio.gubbiotti@unimi.it}
\email{pierandrea.vergallo@unimi.it}
\address[B. van Geemen]{Dipartimento di Matematica ``Federigo Enriques'',
	Universit\`a degli Studi di Milano, Via C. Saldini 50, 20133
	Milano, Italy}
\email{lambertus.vangeemen@unimi.it}

\title[Line geometry, 2nd-order Hamiltonian operators, and quasilinear systems]{Line geometry of pairs of second-order Hamiltonian operators and quasilinear systems}

\date{\today}

\numberwithin{equation}{section}

\begin{document}
\maketitle

\begin{abstract}
   { We demonstrate that a pair consisting of a second-order homogeneous
    Hamiltonian structure in $N$ components and its associated system of
    conservation laws is in bijective correspondence with an alternating
    three-form on a $N+2$-dimensional vector space. Additionally, we show that the
    three-form offers $N+2$ linear equations in the Pl\"ucker coordinates
    that define the associated line congruence. We utilize these results to
    characterize { systems of conservation laws with second-order structure} for $N\leq 4$. }{ We finally comment how to extend this result for $N=6$. }
\end{abstract}



\maketitle


\section{Introduction}

{This work focuses on demonstrating a novel correspondence within the
theory of Hamiltonian structures of Partial Differential Equations (PDEs)
by employing techniques from differential geometry, projective
algebraic geometry, and Lie group actions. Specifically, we provide a
geometric description for pairs consisting of a second-order Hamiltonian
operator and its associated quasilinear PDE system. To the best of our
knowledge, this is the first time such a description has been presented
in the literature.}

The Hamiltonian formalism for PDEs is a fundamental
tool when investigating nonlinear phenomema, since the existence of a Hamiltonian structure implies
that the solutions of the given system possess some form of regularity and
they are connected to symmetries and conserved
quantities~\cite{Olver1986,Magri1978}. A modern approach to Hamiltonian
formalism was presented in the second half of the last
century~\cite{Magri1976}: the main concepts introduced by Hamilton for Ordinary
Differential Equations (ODEs) were extended to Partial Differential Equations
(PDEs) by substituting the Poisson tensor on a symplectic manifold with
integro-differential operators on a space of loops, see e.g.~\cite{mokhov}. As
in the finite dimensional case, Hamiltonian structures endow the {space}
with a deep geometric structure, in both differential and algebraic sense.  {A systematic approach to Hamiltonian operators and geometry
was firstly presented by  Dubrovin and  Novikov in~\cite{DN83},
where the authors introduced the concept of differential-geometric Poisson
structure for quasilinear systems}.

{Now, let us briefly recall the main definitions regarding the Hamiltonian formalism of PDEs}.
Assume we are given two independent variables $t$, the ``evolution'' variable,
and $x$, the ``spatial'' {variable, together with $N$ field variables, denoted by $u=(u^1,\dots, u^N)$. Then, an evolutionary system of $l$ equations has the form:} 
\begin{equation}
    u^i_t=f^i(x,u,u_x, \dots u_{kx}), \qquad i=1,2,\dots, l.
    \label{eq:sysl}
\end{equation}
In particular,  in what follows, we focus on evolutionary quasilinear systems of conservation laws:
\begin{equation}\label{11}
    u^i_t=(V^i(u))_x=V^i_{,j}(u)u^j_x,
    \quad
    V^{i}_{,j} := \frac{\partial V^{i}}{\partial u^{j}},
    \quad i=1,\dots, N,
\end{equation}
{where Einstein's summation convention on repeated indices is used.}
Systems of this form are also called \emph{hydrodynamic-type systems}, and are
very well studied, see e.g.\ \cite{sevennec93:_geomet,tsa1,tsa2}. 

{ Now, }the system~\eqref{eq:sysl} is Hamiltonian if it can be written as:
\begin{equation}
    u^i_t = \mathcal{A}^{ij}\frac{\delta H}{\delta u^j},
    \qquad i=1,2,\dots N,
\end{equation}
{where $\delta/\delta u^j$ is the variational derivate with respect to the field
variable $u^j$, $\mathcal{A}=(\mathcal{A}^{ij})$ is a \emph{Hamiltonian
differential operator}, i.e.\ a it is skew-adjoint $\mathcal{A}^*=-\mathcal{A}$, 
such that the Schouten bracket with itself
vanishes~\cite{GelfandDorfman1981}:}
\begin{equation}
    [\mathcal{A},\mathcal{A}]=0,
    \label{eq:schouten}
\end{equation}
and finally $H$ is a functional:
\begin{equation}
    H=\int{h(x,u,u_x,\dots u_{mx})\, \ud x}.
    \label{eq:Hgen}
\end{equation}

Alternatively, the Hamiltonian property of an operator $\mathcal{A}$ is
expressed in terms of the Poisson bracket conditions~\cite{Magri1976}. That is,
given an operator $\mathcal{A}$, it defines a bracket between two
functionals $F=\int{f\, \ud x}$ and $G=\int{g\, \ud x}$ as:
\begin{equation}
    \{F,G\}_\mathcal{A}=\int{\frac{\delta F}{\delta u^i}\mathcal{A}^{ij}\frac{\delta G}{\delta u^j}\, \ud x}.
    \label{eq:poisson}
\end{equation} 
Then, $\mathcal{A}$ is Hamiltonian if the associated bracket $\{.,.
\}_\mathcal{A}$ is skew-symmetric and satisfies the Jacobi identity, i.e.\ the
bracket~\eqref{eq:poisson} is a Poisson bracket, see also~\cite[Sect.
7.1]{Olver1986}.

{In this paper, } we will consider \emph{local differential operators} of the form:
\begin{equation}
    \mathcal{A}^{ij} = a^{ij\sigma}\partial_{\sigma},
    \label{eq:aijsigma}
\end{equation}
where $a^{ij\sigma}=a^{ij\sigma}(u,u_{x},\dots,u_{Kx})$, { and $\sigma\in\mathbb{N}$ such that $\partial_\sigma=\partial^\sigma /\partial x^\sigma$}. On such operators a grading can
be defined  using the rules:
\begin{equation}
    \deg \partial_{\sigma} = \sigma,
    \quad
    \deg u_{kx} = k.
    \label{eq:grading}
\end{equation}
So, the order of the operator is the maximum of the degrees of all the terms $a^{ij\sigma}\partial_{\sigma}$ in
\eqref{eq:aijsigma}. An operator of this
form is called \emph{homogeneous} if all the terms have the same degree. That is,
an $m$th-order homogeneous operator has the following general form:
\begin{equation}
    \begin{aligned}
        \mathcal{P}^{ij} &=g^{ij}{\partial^m_x}+b^{ij}_k\,u^k_x\,{\partial^{m-1}_x}
        +\left(c^{ij}_k\,u^k_{xx}+c^{ij}_{kl}\,u^k_x\,u^l_x\right){\partial^{m-2}_x}
        +\dots 
        \\
        &+\left(d^{ij}_k\,u^k_{nx}+\dots +d^{ij}_{k_1\dots k_m}\,u^{k_1}_x\,\cdots\, u^{k_m}_x\right) \,.
    \end{aligned}
    \label{hhom}
\end{equation}
where ${g^{ij}}$, $b^{ij}_k$, $c^{ij}_k$, $c^{ij}_{kl}$,\ldots depend on the field variables
$u^1,\dots, u^N$. Homogeneous differential operators were systematically studied
for the first time
in~\cite{DN83}, where the first-order case was considered, while the general
expression for the $m$th-order operator was presented in~\cite{DubrovinNovikov:PBHT}.
The interested reader can see the review paper~\cite{mokhov} for more details
on the topic. 

Second- and third-order operators were investigated independently in details by
Doyle in~\cite{doyle} and Pot\"emin in~\cite{potemin}. In their papers, they found explicitly the
necessary and sufficient conditions for such operators to be Hamiltonian and
proved that there exists a change of  variables  such that
$\mathcal{P}$ can be re-written into the following form
\begin{equation}\label{2}   
\mathcal{P}^{ij}=\partial_x\circ \mathcal{Q}^{ij} \circ \partial_x,
\end{equation} 
where $\mathcal{Q}^{ij}$ is a homogeneous operator of order $0$ and $1$ respectively.
The above canonical form is consequently known as \emph{Doyle-Pot\"{e}min form}
of the operator. Recent developments in this direction show how this canonical
form is typical of a large number of Hamiltonian operators~\cite{LorShaVit1} where
the homogeneous operator $\mathcal{Q}$ is of arbitrary order $d\geq0$.

In recent years, homogeneous Hamiltonian operators have been studied by { applying geometric approaches}, coming both from differential and algebraic geometry,
see~\cite{VerVit2,FPV14,FPV16,FPV17:_system_cl,LorShaVit1}. For instance, in
the non-degenerate case ($\det g\neq0$) it was shown that the leading
coefficient $g^{ij}$ is invariant under projective transformations of the field
variables and the whole operator is invariant under projective-reciprocal
transformations of the independent variables $t$ and $x$, by Ferapontov, Pavlov
and Vitolo for the third-order case~\cite{FPV14}, and by Vitolo and one of the
present authors for the second-order case~\cite{VerVit2}.  This
projective-reciprocal invariance has revealed a \emph{deeper geometric
interpretation} of the Hamiltonian operators and corresponding systems of first order
PDEs in terms of Monge metrics and
alternating three-forms for the third- and second-order cases respectively.

\subsection{Content and structure of the paper}

In this paper we focus on quasilinear systems of first-order PDEs, also known
as \emph{hydrodynamic-type systems}~\cite{sevennec93:_geomet}, admitting a
Hamiltonian structure with a second-order homogeneous operator. We will present
a projective geometric interpretation of the systems and classify them in terms
of projective-reciprocal transformations. {Moreover, we show that the pair
\emph{operator-system} is associated to an alternating three-form  that
uniquely defines it}.  We finally present a direct connection between the
components of the underlying three-form and the coeffiecients of the linear
system satisfyied by the Pl\"ucker coordinates of an associated line
congruence, thus extending the results of~\cite{VerVit2}.

The paper is structured as follows.  In \Cref{sec1}, we review the general
geometry of systems of conservation laws admitting Hamiltonian structure with a
second-order operator. \Cref{sec3} is the core of this paper,  where we present
a bijective correspondence between the pairs formed by a second-order structure
together with the associated system of conservation laws and alternating
three-forms in the $N+2$ projective space. Moreover, here we prove that such systems possess projective-reciprocal
invariance.
{In \Cref{sec6}, we show that in our setting the congruences of lines introduced by Agafonov and Ferapontov~\cite{agafer1,agafer2} for systems of conservation
laws 
is nothing but the annihilator set of lines of the three-form we built in the previous
Section.
We use these results to give a classification of systems with second-order structure in \Cref{sec5} up to $N=4$ and a brief discussion for the case $N=6$}. Finally,
in \Cref{concl}, we present some conclusions and an outlook on this topic.

\section{Homogeneous second-order Hamiltonian structures}\label{sec1}

Before entering the core of the paper, we recall some facts on the geometry of systems of conservation laws
admitting Hamiltonian structure with a second-order operator,
see ~\cite{VerVit2} for further details. The most general operator of this type has the form:
\begin{equation}
    \label{1}
    \mathcal{P}^{ij}=g^{ij}\partial_x^2+b^{ij}_ku^k_x\partial_x+c^{ij}_ku^{k}_{xx}
  +c^{ij}_{kh}u^k_xu^h_x,
\end{equation}
{ where, as mentioned before, all the coefficients depend on the field variables only. In addition, one can show that $g^{ij}$ transforms as a $(2,0)$ tensor whereas $b^{ij}_k$, $c^{ij}_k$ and $c^{ij}_{kl}$ do not transform in a tensorial way. The interested reader can see \cite{Ferguson2008} for explicit formulas (page 470, formulas (6)). }

{  Let us now observe that the skew-adjointness of the operator directly implies $g^{ij}=-g^{ji}$ so that if $g^{ij}$ is
non-degenerate (i.e.\ $\det g \neq 0$) then necessarily
$N=2h$.  Moreover, in the non-degenerate case, the symbols $-g_{li}b^{ij}_k$ and $-g_{li}c^{ij}_k$ follow the transformation rule of a covariant connection.}

{Finally, the  Doyle-Pot\"emin  canonical
form~\eqref{2} of~\eqref{1} is such that $\mathcal{Q}^{ij}=g^{ij}$ is an operator of degree 0 and} defining $g_{ij}={(g^{ij})}^{-1}$ we can show
that~\cite{doyle,potemin}:
\begin{equation}\label{3}
g_{ij}=T_{ijk}u^k+g_{ij}^0,
\end{equation} 
where $T$, $g^0$ are totally skew-symmetric tensors whose components are
constants. So, we can present the two-forms:
\begin{equation}
    g=g_{ij}\ud u^i\wedge \ud u^j,
    \quad 
    g^{0}=g_{ij}^0\ud u^i\wedge \ud u^j,
    \quad i<j,
    \label{eq:gform}
\end{equation}
and the alternating three-form:
\begin{equation}
    T=T_{ijk}\ud u^i\wedge \ud u^j\wedge \ud u^k,
    \quad
    i<j<k,
    \label{eq:Tform}
\end{equation}
in a real or complex vector space of dimension $N$.  Therefore, the number of
independent components of the two-forms $g$ and $g^{0}$ is $N(N-1)/2$, and the
number of independent components of the alternating three-form $T$ is
$N(N-1)(N-2)/6$. In summary, every alternating two-form of type~\eqref{3}
defines a unique homogeneous Hamiltonian operator of second order.  

{As previously mentioned}, projective algebraic geometry plays a key
role in investigating homogeneous Hamiltonian operators of both order 2 and 3.
{For this reason,} let us consider the $N$-dimensional projective space
$\mathbb{P}^N=\mathbb{P}(\mathbb{K}^{N+1})$,
$\mathbb{K}=\mathbb{R},\mathbb{C}$, and let $u^1, \ldots ,u^N,u^{N+1}$ be the
coordinates on $\K^{N+1}$. Following~\cite{VerVit2}, we can define a
homogeneous version $G$ of $g$ in these coordinates, such that
equation~\eqref{3} becomes
\begin{equation}
    G_{ij}=T_{ijk}u^k+g^0_{ij}u^{N+1}.
    \label{eq:Gdef}
\end{equation}
{We then define an alternating three-form $\tilde{T}$ on a $\K$-vector space of dimension $N+1$  as follows:
\begin{equation}
    \label{5}
    \tilde{T}_{ijk}=
    \begin{cases}T_{ijk} \qquad i,j,k\neq N+1,\\
    +g^0_{ij}\qquad k=  N+1, \\
    -g^0_{ik}\qquad j=N+1,\\
    +g^0_{jk}\qquad i=N+1.
    \end{cases}.
\end{equation}
}

This construction implies the following result:
 
\begin{theorem}[\cite{VerVit2}] \label{th:corresp}
  There is a bijective correspondence between the leading coefficients of
  second order homogeneous Hamiltonian operators in Doyle-Pot\"{e}min form and
  the three-forms $\tilde{T}$. Moreover, the bijective correspondence is
  preserved by projective reciprocal transformations up to a conformal factor.
\end{theorem}

Finally, introducing potential coordinates $b^i_x=u^i$, the operator~\eqref{1}
takes the following simple form:
\begin{equation}
    \label{finop}
    \mathcal{P}^{ij}=-g^{ij}.
\end{equation}
Then, for a system of the form~\eqref{11}, the compatibility conditions to be
Hamiltonian with a second order Hamiltonian operator~\eqref{finop} are
expressed by the following theorem:
\begin{theorem}[\cite{VerVit1, VerVit2}] 
    The necessary
conditions for a second-order homogeneous Hamiltonian operator $\mathcal{P}$~\eqref{finop} to be a Hamiltonian operator
for a quasilinear system of first-order conservation laws~\eqref{11} are
\begin{subequations}\label{eq:37}
  \begin{gather}
    \label{eq:451}
    g_{qj}V^j_{,p} + g_{pj}V^j_{,q} = 0,
    \\
    \label{eq:38}
    g_{qk}V^k_{,pl} + g_{pq,k}V^k_{,l} + g_{qk,l}V^{k}_{,p}= 0.
  \end{gather}
\end{subequations}
\end{theorem}

Note that conditions \eqref{eq:37} are algebraic in $g_{ij}$ and they can be
explicitly solved for unknown $V^i$. Indeed, the fluxes $V^i$ satisfying
\eqref{eq:37} have the form
\begin{equation}
    \label{sysm}
    V^i = g^{ij}W_j,
    \qquad \text{where}\quad 
    W_j=A_{jl}u^l+B_j.
\end{equation}
Here $A_{ij}=-A_{ji}$, $B_i$ are arbitrary constants. The interested reader can
see~\cite[Theorem 11]{VerVit2}. 

Solving system~\eqref{eq:37} reveals also an
inner mutual relation between the operator and the Hamiltonian system. In
particular, we indicate with $(\mathcal{P},V)$ the pair
\emph{operator-system} { and we} denote the space of the pairs operator-system
in $N$ components by $\mathcal{Y}_{N}$.

\begin{remark}
    We remark that the fluxes $V^i$ are rational functions whose numerator is a
    polynomial of degree $N/2=h$ in $u$, and the denominator is $\Pf(g)$, the
    Pfaffian of $g$, see~\cite{Muir:TtD60}.  Indeed,  the inverse matrix of
    $g_{ij}$ has rational functions entries where the numerator has degree
    $(N-2)/2$ in $u$ and the denominator is $\Pf(g)$, whose degree in $u$ is at
    most $N/2$. 
    \label{rem:degrees}
\end{remark}

    { \begin{remark}
        Following \cite[Proposition 18]{VerVit2}, we can also comment that if a Hamiltonian structure of the system exists, the corresponding Hamiltonian functional takes the form 
        \begin{equation}
         H=-\int{\left[\frac{1}{2}A_{sl}b^l_xb^s+B_sb^s \right] \ud x } ,
        \end{equation}
        where $b^i_x=u^i$, so that a non-locality in the field variables appears. Explicitly, the system reads as
        \begin{equation}
            b^i_t=g^{ij}\, \frac{\delta H}{\delta b^j}=g^{ij}\left(A_{js}b^s_x+B_j\right).
        \end{equation}
    \end{remark}}

\section{Alternating three-forms and projective-reciprocal {invariance of the system}} \label{sec3}

Based on the results we recalled in the previous section, here we
prove that the pair \emph{operator-system} $(\mathcal{P},V)$ in $N$ components
is in bijective correspondence with alternating three-forms in $N+2$ dimensions, { finally showing
 the projective-reciprocal invariance of the systems.}

\subsection{Equivalence of the pairs $(\mathcal{P},V)$ with alternating three-forms}

The results of~\cite{VerVit2} can be heuristically explained by saying that a
second-order homogeneous Hamiltonian operator $\mathcal{P}$, see~\eqref{3}, in
$N$ components is in bijective correspondence with an alternating three-form
$\tilde{T}$ on $\mathbb{K}^{N+1}$.

Now, we prove that the associated system of conservation laws~\eqref{11} can be
incorporated together with its operator in an alternating three-form on
$\mathbb{K}^{N+2}$. This is the content of the following theorem:

\begin{theorem}\label{thm11}
    There exists a correspondence between the pair $(\mathcal{P},V)$ of the
    second-order operator and { its} associated systems in $N$ components and
    three-forms in $N+2$ components.  Explicitly, there exists a bijective map
    $\Phi\colon \Lambda^3\mathbb{K}^{N+2}\rightarrow \mathcal{Y}_{N}$ defined
    as:
    \begin{align}
       \Lambda^3\mathbb{K}^{N+2}\ni ( \omega_{ijk}) 
       \mapsto \left(\omega_{ijk}u^k+\omega_{ij\,N+1}, g^{is}(\omega_{ij\, N+2}u^j+\omega_{i\, N+1\, N+2})\, \right)
       \in \mathcal{Y}_{N},
    \end{align}
    with inverse 
    $\Phi^{-1}\colon\mathcal{Y}_{N} \rightarrow \Lambda^3\mathbb{K}^{N+2}$ defined
    as:
    \begin{align}
        \mathcal{Y}_{N}\ni(\mathcal{P},V)
        \mapsto
        \Omega=
        \tilde{T}_{ijk} + A\wedge \ud u^{N+2}
        +B \wedge \ud u^{N+1}\wedge \ud u^{N+2}\in\Lambda^{3}\K^{N+2},
    \end{align}
    where $\tilde{T}\in\Lambda^{3}\K^{N+1}$ is defined in equation~\eqref{5},
    and $A\in\Lambda^{2}\K^{N}$, $B\in\Lambda^{1}\K^{N}\simeq\K^{N}$ are the
    constants appearing in equation~\eqref{sysm}.
\end{theorem}

\begin{proof}
        Let us consider an alternating three-form $\Omega
        \in\Lambda^{3}\mathbb{K}^{N+2}$ with components $\omega_{ijk}$.
        Following~\cite[Theorem 2]{VerVit2} let us define 
        \begin{equation}
            g_{ij}=\omega_{ijk}u^k+\omega_{ijN+1}, \qquad i,j=1,\dots N.
        \end{equation}
        Note that this equation uniquely defines a second-order homogeneous
        Hamiltonian operator $\mathcal{P}_\omega$ in flat
        coordinates~\eqref{3}. Moreover, let us define $g^{ij}=(g_{ij})^{-1}$
        and then
        \begin{equation}
            V^i_\omega=g^{is}\left(\omega_{sjN+2}u^j+\omega_{sN+1N+2}\right), 
            \qquad i=1,\dots N.
        \end{equation} 
        By~\eqref{sysm},  this covector satisfies conditions~\eqref{eq:37}.
        Therefore,  the couple $(\mathcal{P}_\omega,V_\omega)$ is a compatible
        pair operator-system.  

        Viceversa, let us consider a pair $(\mathcal{P},V)\in\mathcal{Y}_{N}$. The
        bijection between operators and three-forms $\tilde{T}\in \Lambda^2
        \K^{N+1}$ has been already proved (see \Cref{th:corresp}). Moreover, by
        solving conditions~\eqref{eq:37} we obtain that there exist an
        alternating two-form $A\in\Lambda^{2} \K^{N}$ and a $1$-form
        $B\in\Lambda^1\K^N$ such that $V$ is as in~\eqref{sysm}. Let us now
        define by direct construction of the three-form
        $\omega_{(\mathcal{P},V)}$:
        \begin{subequations} 
            \begin{align}
               \omega_{ijk} &=\tilde{T}_{ijk}, \qquad i,j,k=1,\dots, N+1,
               \\
               \omega_{ijN+2} &=A_{ij},\qquad i,j=1,\dots N,\; k=N+2,
               \\
               \omega_{iN+1N+2} &=B_i,\qquad i=1,\dots N, \; j=N+1,k=N+2,
        \end{align}
        \label{eq:omcoeff}
        \end{subequations}
       which clearly defines a unique alternating three-form
       $\omega_{(\mathcal{P},V)}\in\Lambda^3\K^{N+2}$.
\end{proof}

The previous theorem can be interpreted as a decomposition of the
exterior algebra $\Lambda^3\mathbb{K}^{N+2}$ as follows:
\begin{equation}
    \Lambda^3\K^{N+2}
    =
    \Lambda^3\K^{N+1}\oplus\Lambda^2\K^{N}\oplus\Lambda^1\K^{N},
    \label{eq:decompNppN}
\end{equation}
where $\omega$ corresponds to $(\tilde{T},A,B)$.
Analogously~\cite[see discussion in Section 2.2]{VerVit2} makes use of the following result:
\begin{equation}
    \Lambda^3\K^{N+1}
    =
    \Lambda^3\K^{N}\oplus\Lambda^2\K^{N}.
    \label{eq:decompNppNvervit}
\end{equation}
For an identification with the various $k$-forms we refer to Table
\ref{tab:tabid}.
The decompositions in equations \eqref{eq:decompNppN} and \eqref{eq:decompNppNvervit}
lead to a dimension count: the last entries of Table \ref{tab:tabid} sum to
\begin{equation}
\begin{pmatrix}
    N+2
    \\
    3
\end{pmatrix}=\dim \Lambda^3\K^{N+2},
\quad
\begin{pmatrix}
    N+1
    \\
    3
\end{pmatrix}=\dim \Lambda^3\K^{N+1}.
\end{equation}
{\begin{remark}
   From an analogue point of view, the previous result can be interpreted as a correspondence between three-forms in $N+2$ dimensions and systems admitting a second-order Hamiltonian structure. Furthermore, such a correspondence carries the bijection between three-forms in $N$ dimensions and the corresponding Hamiltonian operators of second order.

    In Section \ref{sec5}, we will show that this result implies that
    classifying systems of conservation laws with second-order structure is
    equivalent to classify three-forms under the action of the special projective
    group.
\end{remark}}
\begin{remark}
    We remark that the last two summands in \eqref{eq:decompNppN} can be joined   to define a single alternating two-form in $\K^{N+1}$ as follows
    \begin{equation}
        \tilde{A}=A_{ij}\, \ud u^{i}\wedge \ud u^j+B_s \, \ud u^s\wedge \ud u^{N+1},
        \label{eq:Atdef}
    \end{equation}
    where:
    \begin{equation}
        \tilde{A}_{ij} =
        \begin{cases}
            A_{ij} & 1\leq i,j\leq N,
            \\
            B_{j} & i=N+1,\,1\leq j\leq N,
            \\
            -B_{i} & 1\leq i\leq N,\,j=N+1.
        \end{cases}
    \end{equation}
    Note that $\tilde{A}_{ij}$ transforms as a $(0,2)$-tensor. 
\end{remark}

\begin{table}[hbt]
    \centering
    \begin{tabular}{ccc}
        \toprule
        Space &  Corresponding form & \# of independent components
        \\[2pt]
        \midrule
        $\Lambda^3 \K^{N}$ & $T$ & $\displaystyle\frac{N(N-1)(N-2)}{6}$
        \\[8pt]
         $\Lambda^2\K^N$ & $g^0$ & $\displaystyle\frac{N(N-1)}{2}$
        \\[8pt]
        $\Lambda^3 \K^{N+1}$ & $\tilde{T}$ & $\displaystyle\frac{(N+1)N(N-1)}{6}$
        \\[8pt]
        $\Lambda^2\K^{N}$ & $A$ & $\displaystyle\frac{N(N-1)}{2}$
        \\[8pt]
        $\K^N$ & $B$ & $N$
        \\[8pt]
        $\Lambda^2\K^{N+1}$ & $\tilde{A}$ & $\displaystyle\frac{(N+1)N}{2}$
        \\[8pt]
        $\Lambda^3\K^{N+2}$ & $\Omega$ & $\displaystyle\frac{(N+2)(N+1)N}{6}$
        \\[8pt]
        \bottomrule
    \end{tabular}
    \caption{Relation among the forms appearing in
    formulas~\eqref{5},~\eqref{eq:omcoeff}, and~\eqref{eq:Atdef}.}
    \label{tab:tabid}
\end{table}

\subsection{Projective-reciprocal invariance}

In~\cite{VerVit2}, the authors proved that the leading coefficient $g^{ij}$ of a
second-order homogeneous Hamiltonian operator is invariant up to a conformal
factor under projective transformations of the field variables:
\begin{equation} 
    \rho:\K^{N}\longrightarrow \K^{N}, \quad \label{trans}
    \rho^i(u)=\tilde{u}^{i}:=\frac{a^i_lu^l+a^i_{N+1}}{a^{N+1}_lu^l+a^{N+1}_{N+1}},
\end{equation}
where $a=(a_{j}^{i})\in\PGL(N+1,\K) = \GL(N+1,\K)/\left\{cI \,\middle|\,
c\in\K\setminus\{0\}\right\}$.  In particular, defining:
\begin{equation}
    A(u):=a^{N+1}_k u^k+a^{N+1}_{N+1},
    \label{eq:Au}
\end{equation}
the two-form $g_{ij}$ transforms into $\bar{g}_{ij}$ under the pullback $\rho^*(g)$ as
follows:
\begin{equation}
    \bar{g}_{ij}\ud\tilde{u}^i\wedge d\tilde{u}^j = \frac{1}{A(u)^3}\,g_{kl} \ud u^k \wedge \ud u^l,
\end{equation}
where $g$ and $\bar{g}$ have the same structure as in equation~\eqref{3},
see~\cite[Corollary 5]{VerVit2}. Thus the leading coefficient of a second-order homogeneous Hamiltonian
operator is form invariant with respect to the action of the group of
projective transformations $\PGL(N+1,\K)$.

Moreover, in Theorem 6 of the same paper, it was proved that the whole operator
in Doyle-Pot\"emin form~\eqref{2} is preserved under the action of
projective-reciprocal transformations { involving the $x$-variables only}. We briefly recall that a reciprocal
transformation is a non-local change of independent variables of the following
form:
\begin{subequations}\label{rec}
\begin{align}
    \ud\tilde{x}&=\left(\alpha^0_iu^i+\alpha^0_0\right)\ud x+\left(\alpha^0_iV^i+\beta\right)\ud t,
    \\
    \ud\tilde{t}&=\left(\beta_iu^i+c\right)\ud x+\left(\beta_iV^i+d\right)\ud t,
\end{align}\end{subequations}
where $\alpha_i^0$, $\alpha_0^0$, $\beta$, $\beta_i$, $c$, and $d$ are arbitrary constants.
For a whole explanation of the theory of reciprocal transformations in the context
of systems of conservation laws we refer to \cite{FerPav1}. The combination of \eqref{trans}
and \eqref{rec} is what is called a projective-reciprocal transformation.

The following result holds:

{\begin{theorem}
    Systems of conservation laws $u^i_t=(V^i(u))_{,x}$ possessing second-order
Hamiltonian formulation are invariant in form under projective-reciprocal transformations
(\eqref{trans}-\eqref{rec}).
    \label{thm:prinv}
\end{theorem}}

Before proving the Theorem we present the following Lemmas:

\begin{lemma}
    The covector $W_i$ is invariant in form under projective transformations \eqref{trans} up to a conformal factor.
    \label{lem:proj}
\end{lemma}

\begin{proof}
    Let us firstly observe that in the transformed frame
    \begin{equation}
        \bar{W} = \bar{W}_i \ud\bar{u}^i = 
        \left(\bar{A}_{is} \bar{u}^s+\bar{B}_i\right)\ud\bar{u}^i.
    \end{equation}
    Then by applying the transformation rule~\eqref{trans}, we have:
    \begin{equation}
        \ud\bar{u}^i=\frac{A(u)a^i_s\ud u^s-\left(a^i_su^s+a^i_{N+1}\right)a^{N+1}_l\ud u^l}{A(u)^2},
    \end{equation}
    where $A(u)$ is defined in equation~\eqref{eq:Au}. So, using the
    skew-symmetry of $\bar{A}_{ij}$ we obtain the following relations:
    \begin{align}\begin{split}
 &   a^i_lu^l\bar{A}_{ij}a^j_su^s=0\qquad
    a^i_{N+1}\bar{A}_{ij}a^j_{N+1}=0\\ 
  &     -a^i_{N+1}\bar{A}_{ij}a^j_sa^{N+1}_l=a^j_{N+1}\bar{A}_{ij}a^j_sa^{N+1}_l\end{split}
    \end{align}
    which give:
    \begin{align}\begin{split}
       &A_{ls}=\frac{1}{A(u)^2}\left[a^i_s\bar{A}_{ij}a^j_l-\bar{B}_ia^i_sa^{N+1}_l+\bar{B}_ia^i_la^{N+1}_s\right] 
       \\
           &B_l=\frac{1}{A(u)^2}\left[a^i_l\bar{A}_{ij}a^j_{N+1}-\bar{B}a^i_{N+1}a^{N+1}_l+\bar{B}_ia^i_la^{N+1}_{N+1}\right].
           \end{split}
       \end{align}
    So, we have:
    \begin{equation}
        W=W_l\ud u^l=\left(A_{ls}u^s+B_l\right)\ud u^l,
    \end{equation}
    i.e.\ $W$ has the same form of $\bar{W}$. Note that the skew-symmetry is
    also preserved, i.e.\ $A_{ij}=-A_{ji}$.
\end{proof}

\begin{lemma}
    A hydrodynamic-type system of the form \eqref{sysm} is invariant under
    the inversion of independent variables:
    \begin{equation}
        \label{eq:xtexch}
        \ud\bar{x} = \ud t, \quad
        \ud\bar{t} = \ud x.
    \end{equation}
    \label{lem:exchange}
\end{lemma}
     
\begin{proof}
    The proof of this lemma is carried out with the same technique as \cite[Theorem 4]{FPV17:_system_cl}. We start noting that
    the exchange of independent variables implies the following transformation
    on the dependent ones:
    \begin{equation}\label{change}
        \bar{u}^i=V^i\qquad \bar{V}^i=u^i.
    \end{equation}
    Then, we claim that the transformed system is still Hamiltonian with
    a second-order Hamiltonian structure given by the following two-form:
    \begin{equation}
        \bar{g}_{ij}=g_{is}V^s_j.
    \end{equation}
    That is, we have to prove that in the new variables $\bar{g}$ has
    the same canonical form as $g$, see equation \eqref{3}.

    This follows from:
    \begin{equation}
        V^i_j=\left(V^i\right)_{,j}=g^{is}\left(A_{sl}+T_{slk}V^k\right)   
    \end{equation}
    which implies
    \begin{equation}
        \bar{g}_{ij}=g_{is}V^s_j=g_{is}g^{sl}\left(A_{lj}+T_{ljk}V^k\right)=T_{ijk}V^k+A_{ij}
        =T_{ijk}\bar{u}^k+A_{ij}.
    \end{equation}
    That is, $\bar{g}$ has the same form as in equation \eqref{3}, because it is 
    skew-symmetric, linear in $\bar{u}$ with the identification $\bar{T}_{ijk}=T_{ijk}$ and $\bar{g}^0_{ij}=A_{ij}$, where $\bar{T},\bar{g}^0$ are constant and skew-symmetric with respect to any exchange of indices.

    Finally, we prove that the structure of the system is preserved, i.e. it has
    the following form:
    \begin{equation}
        \bar{V}^i = \bar{g}^{ij}\bar{W}_j = 
        \bar{g}^{ij}(\bar{A}_{js}\bar{u}^s+\bar{B}_j).
        \label{eq:Vtd}
    \end{equation}
    By \eqref{sysm}, we have that
    \begin{equation}
    T_{isk}u^kV^s+g^0_{is}V^s=A_{is}u^s+B_i,
    \end{equation}
    and using \eqref{change} the transformed relation is:
    \begin{equation}
{T}_{isk}\bar{V}^k\bar{u}^s+g^0_{is}\bar{u}^s={g}_{is}^0\bar{V}^s+B_i.
    \end{equation}
    So, equation \eqref{eq:Vtd} follows setting
    $\bar{T}_{ijk}=T_{ijk}$, $\bar{g}_{ij}^0=A_{ij}$, $\bar{A}_{ij}=g^0_{ij}$ 
    and $\bar{B}_i=-B_i$. This ends the proof of the Lemma.
\end{proof}

\begin{proof}[Proof of Theorem \ref{thm:prinv}]
    We notice that proving the invariance of the system under projective-reciprocal transformations is equivalent to
    prove its projective invariance and its invariance under complete reciprocal transformations \eqref{rec}.

    The first statement follows from Lemma \ref{lem:proj}, the explicit form of $V$ given in
    equation \eqref{3}, and the fact that $g$ is invariant in form under projective 
    transformations, see \cite[Corollary 5]{VerVit2}.

    To prove the reciprocal invariance, we recall that a general reciprocal
    transformation \eqref{rec} can be viewed as
    \begin{equation}
        (x\text{-transformation})\circ (x\leftrightarrow t)\circ(x\text{-transformation}),
    \end{equation}
    where by $x$-transformation we mean a reciprocal transformation changing the variable $x$ only,
    and by $x \leftrightarrow t$ we mean the independent variables exchange, see for instance
    \cite{FPV17:_system_cl}.
    
    In \cite[Theorem 6]{VerVit2}, the invariance of Hamiltonian formalism systems under $x$-translations 
    has already been proved. The invariance under independent variable exchange was shown in
    Lemma \ref{lem:exchange}, {and this concludes} the proof of the Theorem.
\end{proof}

\section{Linear line congruences and hydrodynamic type systems}\label{sec6}

{We now want to deeply understand the geometric interpretation of operators and systems  in the context of classical projective geometry. To this aim, we recall that} in~\cite{agafer1,agafer3, agafer2}, the authors presented an interpretation of
hydrodynamic-type systems of conservation laws~\eqref{11}, involving the
classical theory of congruence of lines in the projective space. Using this
theory, some basic concepts of homogeneous quasilinear systems such as shocks,
rarefaction curves and linear degeneracy  are treated by means of geometric
properties of the associated algebraic variety. 

In particular, it has been shown that to every conservative hydrodynamic-type
system one can associate a congruence of lines
\begin{equation}\label{linc}
  y^i=u^iy^{N+1}+V^iy^{N+2}, \qquad i=1,2,\dots N,
\end{equation}
in an auxiliary projective space $\mathbb{P}^{N+1}$ with homogeneous
coordinates $[y^1:\ldots : y^{N+2}]$. So,  for each field variable
$u=(u^1,\ldots,u^N)$ one  considers the line $L_u\subset \mathbb{P}^{N+1}$
spanned by the two points:
\begin{equation}
    P_u\,=\,[u^1:\ldots:u^N:1:0],
    \quad Q_u\,=\,[V^1:\ldots:V^N:0:1].
    \label{eq:PuQu}
\end{equation}
Consequently, the Pl\"ucker coordinates $p^{ij}=p^{ij}(u)$ of
$L_u$ are defined as the determinants of all $2\times 2$ submatrices of 
\begin{equation}
    \begin{pmatrix}
        u^1&u^2& \ldots&u^N&1&0\\
        V^1&V^2& \ldots&V^N&0&1
    \end{pmatrix}
\end{equation} 
or explicitly:
\begin{align}\begin{split}\label{plu}
&p^{kl}\,=\,u^kV^l-u^lV^k,\qquad p^{r,N+1}\,=\,-V^r,\\ &p^{r,N+2}\,=\,u^r,\qquad p^{N+1,N+2}\,=\,1~.\end{split}
\end{align}
Notice that the Pl\"ucker coordinates define an embedding of the Grassmannian $\Gr(2,N+2)$ in a projective space $\PP^M$ with
$M=\binom{N+2}{2}-1$.

Finally, recall that a congruence~\eqref{linc} is said to be \emph{linear} if its
closure in $\Gr(2,N+2)$ is defined by linear relations between the Pl\"ucker
coordinates~\eqref{plu} of its lines. {We observe that this is exactly our case, indeed} using~\eqref{sysm}, one can invert $g^{ij}$, apply~\eqref{3} and obtain $N$
linear equations between the Pl\"ucker coordinates~\eqref{plu} of the
lines~\eqref{linc}:
\begin{equation}
    \label{eqs}
    \begin{split}
    \frac{1}{2}T_{jkl}\left(u^lV^k-u^kV^l\right)+g^0_{jk}V^k-A_{jl}u^l-B_j=0, \quad j=1,2,\dots N,
  \end{split}  
\end{equation}
see~\cite[Theorem 11]{VerVit2}.

{For sake of simplicity,  we now study in detail the line congruence associated to an example
coming {from the $2$-component case}.}

\begin{example}
    \label{exaN2}
    We consider the case $N=2$. Then, by applying the transformation
    $u^{3}\mapsto u^{3}/{g_{12}^0}$ the {homogeneized} operator reduces to
    \begin{equation}
        \mathcal{P}=\begin{pmatrix}
            0&1\\-1&0
        \end{pmatrix}\partial_x^2,
    \end{equation}
   { see also \cite[Section 2.3.1]{VerVit2}}. Here, all the other terms
    in~\eqref{1} vanish due to the constant entries of the leading coefficient
    $g^{ij}$. Now, from $V^i=g^{is}(A_{sl}u^s+B_s u^{N+1})$ we obtain
    \begin{equation}
        V^1=A_{12}u^1-B_2u^{3},
        \qquad
        V^2=A_{12}u^2+B_1u^{3}
    \end{equation}
    so that the line $L_u\subset\PP^3$, with $u=(u^1:u^2:u^3)\in\PP^2_u$, is
    spanned by:
    \begin{subequations}
        \begin{align}
            P_u &= \left[u^1 : u^2: u^{3}:0\right]
            \\
            Q_u &=\left[A_{12}u^1-B_2u^{3}:A_{12}u^2+B_1u^{3}:0:u^{3}\right]~.
        \end{align}
    \end{subequations}
    The Pl\"ucker coordinates of $L_u$ are obtained from the matrix
    \begin{equation}
    r\colon \begin{pmatrix}
    u^1&u^2&u^{3}&0\\
    A_{12}{u^1}-B_2u^{3}&A_{12}u^2+B_1u^{3}&0&u^{3}
    \end{pmatrix}
    \end{equation}
    and they are:   %
    \begin{align}\begin{split}
    &p^{12}=B_1u^1u^{3}+B_2u^2u^{3},  \quad 
    p^{13}=-A_{12}u^1u^{3}+B_2(u^{3})^2,\quad
    p^{14}=u^1u^{3},\\
    &p^{23}=-A_{12}u^2u^{3}-B_1(u^{3})^2, \quad
    p^{24}=u^2u^{3}, \quad 
    p^{34}=(u^{3})^2~.
    \end{split}
    \end{align}
    The line $L_u$ is represented by the point
    $(p^{12}:\ldots:p^{34})\in\PP^5$. Notice that all six $p^{ij}$ are
    multiples of $u^{3}$ and thus the line $L_u$ is also represented by
    \begin{align}\begin{split}
        &p^{12}=B_1u^1+B_2u^2,  \quad
        p^{13}=-A_{12}u^1+B_2u^{3},\quad
        p^{14}=u^1,\\
        &p^{23}=-A_{12}u^2-B_1u^{3}, \quad
        p^{24}=u^2, \quad
        p^{34}=u^{3}~.
    \end{split}
    \end{align}
    Since all $p^{ij}$ are linear in the $u^k$, the image of $\PP^2_u$ in
    $\Gr(2,4)\subset\PP^5$ is again a linearly embedded $\PP^2$.  It is
    well-known that there are two types of such planes in the Grassmannian, one
    type parametrizes all the lines in $\PP^3$ through a given point and the
    other type parametrizes all lines in $\PP^3$ contained in a plane. All the
    lines $L_u$ in fact contain a point $R$:
    \begin{equation}
        L_u\,=\,\langle P_u,Q_u\rangle\;=\;\langle P_u,R\rangle,
        \qquad 
        R\,:=\,\left[-B_2:B_1:-A_{12}:1\right]~.
    \end{equation}
    Thus for $N=2$ the lines associated to a hydrodynamic type system are
    exactly all the lines in $\PP^3$ through a point.

    Now we consider the linear equations satisfied by the Pl\"ucker coordinates
    of the lines $L_u$.  These coordinates are not all linear independent,
    indeed~\eqref{eqs} provides two linear relations, which are the first two
    in~\eqref{3451}, and it is not hard to find two more:
    \begin{subequations}\label{3451}
    \begin{align}\label{eq1}
    p^{13}+A_{12}p^{14}-B_2p^{34}&=0, \\ \label{eq2}
    p^{23}+A_{12}p^{24}+B_1p^{34}&=0\\\label{eq3}
    p^{12}-B_1p^{14}-B_2p^{24}&=0, \\\label{eq4}
    A_{12}p^{12}+B_1p^{13}+B_2p^{23}&=0
    \end{align}
    \end{subequations}
    The four relations~\eqref{3451} are in fact linearly dependent since
    multiplying the first three by $B_1$, $B_2$, and $A_{12}$ respectively and
    summing gives the last equation. The first three equations are independent
    and they define exactly the conguence of lines $L_u$. In fact, the image of
    the four dimensional $\Gr(2,4)$ in $\PP^5$ is defined by the so-called
    Pl\"ucker quadric:
    \begin{equation}
        p^{12}p^{34}\,-\,p^{13}p^{24}\,+\,p^{14}p^{23}\,=\,0~.
        \label{eq:pluckerquadr}
    \end{equation}
    From the first two equations, which are~\eqref{eqs}, we find:
    \begin{equation}
        p^{13}=-A_{12}p^{14}+B_2p^{34},
        \qquad
        p^{23}=-A_{12}p^{24}-B_1p^{34}, 
    \end{equation}
    and substituting in the Pl\"ucker quadric~\eqref{eq:pluckerquadr} we obtain
    \begin{equation}
        p^{34}(p^{12} - B_1p^{14} - B_2p^{24})\,=\,0~.
    \end{equation}
    Thus \eqref{eq1} and \eqref{eq2} define a $\PP^3\subset\PP^5$ that cuts
    the quadric in the union of two planes, one is defined by the two equations
    and $p^{34}=0$ whereas the other is defined by first three equations.  Thus
    the congruence $L_u$ is defined by \eqref{eq1}, \eqref{eq2} and \eqref{eq3}. 
    { Notice that the first two equations by themselves do not suffice to define 
    the congruence associated to a hydrodynamic type system.}

    In~\cite[Example 4.9]{DePoiFaenziMezzettiRanestad17} a three-form on $\K^4$
    is used to define a system of three linear equations
    $p_{12}=p_{13}=p_{23}=0$.  One easily verifies that these define the lines
    in $\PP^3$ that pass through the point $(0:0:0:1)$, similar to the lines in $L_u$
    that pass through $R$.

    Also our equations~\eqref{3451} can be re-written using the alternating
    three form $\Omega= (\omega_{ijk})$ as in~\eqref{eq:omcoeff}, so
    $\omega_{ijk}=-\omega_{jik}$, $\omega_{ijk}=-\omega_{ikj}$ with
    $\omega_{123}=1$, so that $\tilde{T}=\ud u^1\wedge \ud u^2\wedge \ud u^3$,
    $\omega_{124}=A_{12}$, $\omega_{134}=B_1$ and $\omega_{234}=B_2$:
    \begin{subequations}
        \begin{align}
            0&=\omega_{123}p^{23}+\omega_{124}p^{24}+\omega_{134}p^{34}=p^{23}+A_{12}p^{24}+B_1p^{34}
            \\
            0&=\omega_{213}p^{13}+\omega_{214}p^{14}+\omega_{234}p^{34}=-p^{13}-A_{12}p^{14}+B_2p^{34}
            \\
            0&=\omega_{312}p^{12}+\omega_{314}p^{14}+\omega_{324}p^{24} =p^{12}-B_1p^{14}-B_2p^{24}
            \\
            0&=\omega_{412}p^{12}+\omega_{413}p^{13}+\omega_{423}p^{23}=A_{12}p^{12}+B_1p^{13}+B_2p^{23}
        \end{align}
        \label{902}%
    \end{subequations}
    This can be compared to formula~\eqref{eqs}, where one has to recall
    that the 1/2 factor comes from the symmetrisation.
    \hfill$\square$
\end{example}
We now increase the number of components to $N=4$ and analyze the following case. 
\begin{example}\label{exaN4}
    We now consider the case $N=4$. We choose an operator in Doyle-Pot\"{e}min canonical
    form  determined by the following alternating $4\times 4$ matrix
    $g=(g_{ij})$, which does not depend on $u^4$ and which we homogenize with a
    variable $u^5$:
    \begin{equation}
        \label{ese2}
        g\,=\,g(u)\,=\,\begin{pmatrix}
        0&u^3+g^0_{12}u^5&-u^2+g^0_{13}u^5&g_{14}^0u^5\\
        -u^3-g^0_{12}u^5&0&u^1+g_{23}^0u^5&g_{24}^0u^5\\
        u^2-g^0_{13}u^5&-u^1-g^0_{23}u^5&0&g_{34}^0u^5\\
        -g^0_{14}u^5&-g^0_{24}u^5&-g^0_{34}u^5&0
        \end{pmatrix}~.
    \end{equation}
    The fluxes $V^i$ are determined by an alternating complex $4\times 4$
    matrix $A$ and a $B\in\K^4$ as $V=g^{-1}(AU+Bu^5)$ where
    $U:=(u^1,\ldots,u^4)$ as in~\eqref{sysm}.  To find $g^{-1}=(g^{ij})$ we
    observe that the determinant of an alternating matrix is the square of its
    Pfaffian, in this case one has
    \begin{equation}
        \Pf(g)\,=\,u^5\left[g_{14}u^1 + g_{24}u^2 + g_{34}u^3 + (g_{12}g_{34} - g_{13}g_{24} + g_{14}g_{23})u^5\right]~.
    \end{equation}
    Moreover, the inverse of $g$ can be obtained as
    \begin{align}\label{ginv}
    &g^{-1}\,=\,\frac{1}{\Pf(g)}g^\sharp,\qquad
    g^\sharp\,=\,\begin{pmatrix}
     0& -g_{34} & g_{24}& g_{23}\\
    g_{34}&    0& -g_{14}&  g_{13}\\
    -g_{24}&  g_{14}&    0& -g_{12}\\
     g_{23}& -g_{13}&  g_{12}&    0
    \end{pmatrix}~.\end{align}

    The lines associated to these fluxes are the $L_u\subset\PP^5$, with
    $u=(u^1:u^2:u^3:u^4:u^5)\in\PP^4_u$, where $L_u$ is spanned by
    \begin{subequations}
        \begin{align}
            P_u&=\left[u^1 : u^2: u^{3}:u^{4}:u^{5}:0\right]
            \\
            Q_u&=\left[V^1:V^2:V^3:V^4:0:u^{5}\right]~.
        \end{align}
    \end{subequations}
    Since $Q_u\in\PP^5$ we may multiply all coordinates by $\Pf(g)$ so that
    \begin{equation}
        Q_u\,=\,\left[V_\sharp^1:V_\sharp^2:V_\sharp^3:V_\sharp^4:0:\Pf(g)u^{5}\right],
        \qquad
        V_\sharp\,:=\,g^\sharp(AU+Bu^5)~.
    \end{equation}
    All coefficients of $P_u$, $Q_u$ are homogeneous of degree $1$ and $2$
    respectively in the $u^i$.  The fifteen Pl\"ucker coordinates
    $p^{12},\ldots,p^{56}$ of $L_u$ are then homogeneous of degree three in the
    $u^i$. A computations shows that they are all divisible by $u^5$, so the
    point $(p^{12}:\ldots:p^{56})\in\PP^{14}$ defined by $L_u$ has coordinates
    that are homogeneous of degree two in  $u^i$. For example:
    \begin{equation}
        p^{12}\,=\, (-g_{13}A_{14} + g_{14}A_{13})(u^1)^2 
        + \ldots
        +(g_{23}B_4 - g_{24}B_3 + g_{34}B_2)u^2u^5~.
    \end{equation}
    With computer algebra one can study this Pl\"ucker map from
    $\PP^4_u\rightarrow\Gr(2,6)\subset\PP^{14}$, but instead we will now
    consider the equations defining the image.

    The Pl\"ucker coordinates~\eqref{plu} of the lines $L_u$ satisfy six linear
    equations of the form
    \begin{equation}
        \omega_{ijk}p^{jk}=0, \qquad i=1,2,\dots N+2,
    \end{equation}
    where again $\omega=(\omega_{ijk})$ is an alternating three-form.  In the
    following table, the column under $p^{jk}$ lists the coefficients
    $\omega_{ijk}$, $i=1,\ldots,6$ of $p^{jk}$ in the six equations.

    \begin{subequations}\label{eqs22}\begin{equation}
    \begin{array}{rrrrrrrrr}
    &p^{12}&p^{13}&p^{14}&p^{15}&p^{16}&p^{23}&p^{24}&p^{25}\\
    \\
    \ud u^1&0&0&0&0&0&1&0&g^0_{12}\\
    \ud u^2&0&-1&0&-g^0_{12}&-A_{12}&0&0&0\\
    \ud u^3&1&0&0&-g^0_{13}&-A_{13}&0&0&-g^0_{23}\\
    \ud u^4&0&0&0&-g_{14}&-A_{14}&0&0&-g^0_{24}\\
    \ud u^5&g^0_{12}&g^0_{13}& g^0_{14}&0&-B_1&g^0_{23}&g^0_{24}&0\\
    \ud u^6&A_{12}&A_{13}&A_{14}& {B_1}&0&A_{23}&A_{24}&{B_2}
    \end{array}
    \end{equation}
    \vspace{2mm}
    \begin{equation}
        \begin{array}{rrrrrrrr}
    &p^{26}&p^{34}&p^{35}&p^{36}&p^{45}&p^{46}&p^{56}\\\\
    \ud u^1&A_{12}&0&g^0_{13}&A_{13}&g^0_{14}&A_{14}&B_1\\\ud u^2&0&0&g^0_{23}&A_{23}&g^0_{24}&A_{24}&B_2\\
    \ud u^3&-A_{23}&0&0&0&g^0_{34}&A_{34}&B_3\\
    \ud u^4&-A_{24}&0&-g^0_{34}&-A_{34}&0&0&B_4\\
    \ud u^5&-B_2&g^0_{34}&0&-B_3&0&-B_4&0\\
    \ud u^6&0&A_{34}&{B_3}&0&{B_4}&0&0\\
    \end{array}
    \end{equation}\end{subequations}
    The first four equations are those from~\eqref{eqs}.

    The equations~\eqref{eqs22} were analyzed in~\cite[Example
    4.11]{DePoiFaenziMezzettiRanestad17}, they define a subspace
    $\PP^8_\omega\subset \PP^{14}$.  The intersection
    {$X_\omega:=\PP^8_\omega\cap \Gr(2,6)$} is a four dimensional subvariety
    $X_\omega$ of $\Gr(2,6)$ which is isomorphic to (the Segre image of)
    $\PP^2\times\PP^2$.  In that paper one also finds that the first four
    equations define a union of two (irreducible) subvarieties, $X_\omega$ and
    $Y=Y_{\omega,\ud u^5\wedge \ud u^6}$.  The last two equations thus are needed to
    exclude the points in $Y$ that are not in $X_\omega$.

    Since the Pl\"ucker map $\PP^4_u\rightarrow \Gr(2,6)$ is easily seen to
    have degree one onto its image, we conclude that its image is $X_\omega$.
    The Pl\"ucker map thus induces a birational isomorphism
    $\PP^4_u\rightarrow\PP^2\times\PP^2$ (the base locus consists of two skew
    lines) which is similar to the birational isomorphism between $\PP^2$ and
    $\PP^1\times\PP^1$ which blows up two points and contracts the line in
    $\PP^2$ spanned by the base points.

    Consider now the general three-form
    \begin{equation}
        \omega\,=\, (\omega_{ijk})\,=\,\ud u^1 \wedge \ud u^2 \wedge \ud u^3 + \ud u^4 \wedge \ud u^5 \wedge \ud u^6~,
    \end{equation}
    so that the $\omega_{ijk}$ are:
    \begin{equation}
        \omega_{123} \,=\, \omega_{456}\,=\,1,\quad
        \omega_{ijk}\,=\, 0 
        \quad\mbox{for}\quad  (i, j, k) \neq (1, 2, 3), (4, 5, 6).
    \end{equation}
    The $N+2=6$  equations in the Pl\"ucker equations that define
    $X_\omega\subset \Gr(2,6)$ are:
    \begin{equation}
        \pm \omega_{123}p^{12}\,=\,0,
        \quad 
        \pm \omega_{123}p^{13}\,=\,0,
        \quad 
        \pm \omega_{123}p^{23}\,=\,0,
    \end{equation}
    and similarly there are three equations involving $\omega_{456}$. The lines
    with $p_{12}=p_{13}=p_{23}=0$ are those which meet the plane
    $x_1=x_2=x_3=0$.  In fact, such a line is spanned by
    $a=(a_1,\ldots,a_6),\,b=(b_1,\ldots,b_6)\in\K^6$, and the vectors
    $(a_1,a_2,a_3)$, $(b_1,b_2,b_3)\in\K^3$ are linearly dependent. So after
    taking a suitable linear combination of $a$ and $b$ we may assume that
    $a_1=a_2=a_3=0$ and then $a$ lies in $x_1=x_2=x_3=0$. Similarly the lines
    with $p_{45}=p_{46}=p_{56}=0$ are those which meet the plane
    $x_4=x_5=x_6=0$.  Thus any line in $X_\omega$ is spanned by a point
    $P=(x_1:x_2:x_3:0:0:0)$ and a point $Q=(0,0,0,y_1:y_2:y_3)$.
    The Pl\"ucker coordinates of this line are the $x_iy_j$ so $X_\omega\cong\mathbb{P}^2\times\mathbb{P}^2$ and
    $X_\omega$ is embedded via the
    Segre map in the $\mathbb{P}^8\subset\mathbb{P}^{14}$ defined by $p_{12}=p_{13}=p_{23}=p_{45}=p_{46}=p_{56}=0$.
    \hfill $\square$
\end{example}

As observed in the Examples~\ref{exaN2} and~\ref{exaN4} for $N=2,4$, besides
the $N$ linear equations for the Pl\"ucker coordinates given in \eqref{eqs},
these satisfy two more linear equations. The $N+2$ equations one finds can be
written as
\begin{equation}\label{omegaeqs}
    \omega_{ijk}p^{jk}=0,\qquad i=1,2,\dots N+2,
\end{equation}
where the coefficients $\omega_{ijk}$ are constant and skew-symmetric with respect
to any pair of indices.

The results of~\cite{DePoiFaenziMezzettiRanestad17} imply that these linear
equations, for a general alternating three-form $\omega$, define a smooth
subvariety $X_\omega$ of dimension $N$ of $\Gr(2,N+2)$.

We will show in Theorem~\ref{thm:cong} that for any even $N$ the Pl\"ucker
coordinates of the lines $L_u$ satisfy the equations~\eqref{omegaeqs}, where
$\omega$ is the three-form defined in~\eqref{eq:omcoeff}. Thus $X_\omega$ is
birationally isomorphic to $\PP^N_u$, the projective space which parametrizes
the lines $L_u$. In particular, we found an explicit parametrization of
$X_\omega$, which is thus a rational variety.

\begin{remark}
    \label{rem:discrepancy}
    Notice that $\dim\Gr(2,N+2)=2N$ and $\dim\PP^N_u=\dim X_\omega=N$. However,
    there are $N+2$ linear forms vanishing on $X_\omega$ which are linearly
    independent for $N>2$. In fact, in~\cite{DePoiFaenziMezzettiRanestad17} it
    is shown that taking $N$ of these linear forms defines a dimension $N$
    subset with two irreducible components one is $X_\omega$, the other is
    denoted by $Y$. One needs two more equations to find exactly $X_\omega$.
    The two components are very well visible in Example~\ref{exaN2} but in
    general we do not have a good description of the `extra' component $Y$
    defined by the first $N$ equations in terms of Hamiltonian structures. This
    also explains the ``dimensional gap'' between the description of
    second-order Hamiltonian operators of~\cite{VerVit2} in terms of
    alternating three-forms on $\mathbb{K}^{N+1}$, and the projective
    interpretation of hydrodynamic-type systems~\eqref{sysm} of Agafonov and
    Ferapontov~\cite{agafer1,agafer3,agafer2} on the projective space
    $\mathbb{P}^{N+1}$.
\end{remark}

{ Furthermore, we can also write the equations \eqref{omegaeqs} more intrinsically.}
Let us consider $\Omega=(\omega_{ijk})\in\Lambda^{3}\K^{N+2}$, and a line  $L\subset
\mathbb{P}^{N+1}$, i.e.\ a two-dimensional subspace in $\K^{N+2}$.
Recall that the pullback $\Omega^{*}L\in \Lambda^{1}\K^{N+2}$ of $\Omega$ to
$L$ is the contraction w.r.t.\ two indices, so in coordinates:
\begin{equation}
    (\Omega^{*}L)_{i} = \omega_{ijk}p^{jk}, 
    \label{eq:pullback}
\end{equation}
where $p^{jk}$ are the Pl\"ucker coordinates of $L$ in $\Gr(2,N+2)$.
{Then, the \emph{annihilation set of lines} for $\Omega$, denoted by $X_\Omega$,
is defined by those lines $L\subset\mathbb{P}^{N+1}$ whose pullback with respect to $\Omega$
vanishes:}
\begin{equation}
    X_\Omega\,=\,\{[L]\in \Gr(2,N+2):\;\Omega^{*}L\,=\,0\,\}~.
\end{equation}
In coordinates, $\Omega^{*}L=0$ is the set of $N+2$ linear equations in the Pl\"ucker coordinates of $L$
$\omega_{ijk}p^{jk}=0$.
The following result shows that the lines $L_u$ we considered satisfy the equations obtained from $\Omega$:

\begin{theorem}\label{thm:cong}
    The congruence of lines associated to a hydrodynamic-type system with
    second-order homogeneous Hamiltonian structure is the annihilation set
    of lines of the three-form $\Omega$ associated to the operator-system pair
    $(\mathcal{P},V)$ in the sense of Theorem \ref{thm11}.
\end{theorem}

\begin{proof}
{ Let us consider the system of $N$ linear equations \eqref{eqs} in the Pl\"ucker coordinates that are satisfied by the lines $L_u$. These can be rewritten as
    \begin{equation}
        \sum_{k,l=1, \, k<l}^NT_{jkl}p^{kl}+\sum_{k=1}^Ng^0_{jk}p^{k\, N+1}+
        \sum_{k=1}^NA_{jk}p^{k\, N+2}+B_jp^{N+1\, N+2}=0
        \label{02341}
    \end{equation}
    for $i=1,\ldots, N$. } Then using the definition of $\Omega$ \eqref{eq:omcoeff} we obtain:
    \begin{align}\begin{split}
       & \sum_{k,l=1, \, k<l}^N\omega_{jkl}p^{kl}+\sum_{k=1}^N\omega_{jk\, N+1}p^{k\, N+1}+\\&\hphantom{ciaciaoo}+\sum_{k=1}^N\omega_{jk\, N+2}p^{k\, N+2}+\omega_{j\, N+1\, N+2}p^{N+1\, N+2}=0\end{split}
    \end{align}
    that is 
    \begin{equation}
        \sum_{k,l=1\, k<l}^{N+2}\omega_{jkl}p^{kl}=0,  \qquad  j=1,2,\ldots,N+1,
    \end{equation}
    thus proving the first $N$ equations obtained from $\Omega$ are satisfied.

    Let us now derive the last two equations.  At first, let us consider
    $g_{jk}V^k=A_{jl}u^l+B_j$, then by multiplying by $u^j$  we use the
    skew-symmetry to obtain:
\begin{equation}
 (T_{jkl}u^l\,+\,g_{jk}^0)V^ku^j\,=\,A_{jl}u^lu^j+B_ju^j,
\end{equation}
 so that after  using \eqref{plu}, one obtains
\begin{equation}
-\frac{1}{2}g^0_{jk}(u^jV^k-u^kV^j)\,=\, B_ju^j\quad\Longrightarrow\quad
\frac{1}{2}g_{jk}^0p^{jk}\,-\,B_jp^{j\, N+2}\,=\,0.
\end{equation}
Using  \eqref{eq:omcoeff}, the previous expression is
\begin{equation}
    \sum_{k,l=1\, k<l}^N\omega_{N+1\, kl}p^{kl}+\sum_{k=1}^N\omega_{N+1\, j\, N+2}p^{j\, N+2}=0.
\end{equation}
This is in fact the $N+1$-st equation:
\begin{equation}
    \sum_{k,l=1\, k<l}^{N+2}\omega_{N+1\, kl}p^{kl}=0.
\end{equation}

The last relation is obtained similarly, by multiplying  $g_{jk}V^k=A_{jl}u^l+B_j$ by $V^j$
and by using the skew-symmetry:
\begin{equation}
 (T_{jkl}u^l\,+\,g_{jk}^0)V^kV^j\,=\,A_{jl}u^lV^j+B_jV^j
\end{equation}
which implies
$
A_{jl}u^lV^j+B_jV^j\,=\,0$, and finally 
\begin{equation}\frac{1}{2}A_{jl}p^{jl}-B_kp^{kN+1}\,=\,0.\end{equation}
This is 
\begin{equation}
    \sum_{k,l=1\, k<l}\omega_{N+1\, kl}p^{kl}+\sum_{k=1}^N\omega_{N+2\, k\, N+2}p^{k\, N+1}=0
    \end{equation} 
    or equivalently, 
    \begin{equation}
        \sum_{k,l=1\, k<l}^{N+2}\omega_{N+2\, kl}p^{kl}=0.
    \end{equation}
    This is the $N+2$-st equation. Note that by construction the coefficients are totally skew-symmetric and the theorem is proved.
\end{proof}

{ \begin{remark}\label{rem1}

One can observe that if $\tilde{T}$ is non-degenerate (i.e., it defines a
non-degenerate second-order Hamiltonian operator as in \cite{VerVit2}),
then the converse of this statement is also true. Indeed, if
\begin{equation}
    \omega=\tilde{T} + A\wedge du^{N+2} + B\wedge du^{N+1}\wedge
du^{N+2},
\end{equation} and $\tilde{T}$ is non-degenerate, then we define $g_{ij}$
as in \cite{VerVit2}. We invert it and find the following lines:
\begin{equation}
    y^i=u^iy^{N+1}+g^{is}\left(\omega_{s\, l\, N+2}u^l+\omega_{s\, N+1\,
N+2}\right)y^{N+2}.
\end{equation} These lines uniquely define the annihilation set of
$\omega$ (and the associated system).
\end{remark}}

\vspace{3mm}


{In summary, Theorems \ref{thm11} and \ref{thm:cong} state that the
natural geometric structures of hydrodynamic-type systems admitting
Hamiltonian formalism with second-order homogeneous operators are
alternating three-forms in the projective space $\mathbb{P}^{N+1}$ along
with their annihilation sets of lines.}

\section{Classification of {systems of conservation laws systems with second-order Hamiltonian structure}}
\label{sec5}

{
In \cite{agafer1,agafer3}, the authors show that the group of complete reciprocal transformations \eqref{rec} together with the affine transformations of the field variables ${u}$ is isomorphic to the projective group in $\mathbb{P}^{N+1}$. So that, quoting \cite[pag. 156]{agafer3}:
\begin{quote}
\emph{``[\ldots] the
classification of systems of conservation laws up to transformations (8) [\eqref{rec}] and affine changes of ${u}$ is
equivalent to the classification of the corresponding congruences up to projective equivalence.''} 
\end{quote}
 In addition, it is a remarkable fact that the Hamiltonian property is also preserved by these transformations, as proved in Theorem ~\ref{thm:prinv}.  For our purposes, we further point out that Theorem \ref{thm:cong} and Remark \ref{rem1} guarantee that {classifying up to the action of $\PGL(N+2,\K)$ congruences of lines as }
 \begin{equation}
     y^i=u^iy^{N+1}+g^{is}\left(A_{sl}u^l+B_s\right)y^{N+2}, \qquad i=1,2,\dots, N,
 \end{equation}
{ where $g_{ij}=T_{ijk}u^k+g^{0}_{ij}$ as defined before,
 is equivalent to classify three-forms in $\Omega\in  \Lambda^3\mathbb{K}^{N+2}$ in the form}
\begin{equation}\label{dec1}
\begin{array}{rcl}
    \Omega&=&\tilde{T}\,+\,A\wedge \ud u^{N+2}\,+\,B\wedge \ud u^{N+1}\wedge \ud u^{N+2}\\
   &=&\tilde{T}\,+\,\tilde{A}\wedge \ud u^{N+2}~,
\end{array}
\end{equation}
where
$\tilde{T}\in\Lambda^3\mathbb{K}^{N+1}$, $A\in \Lambda^2\mathbb{K}^{N}$,
$B\in\Lambda^1\mathbb{K}^N\cong \K^N$ and $\tilde{A}\in \Lambda^2\mathbb{K}^{N+1}$, satisfying  the two following consistency conditions:
\begin{enumerate}
    \item the alternating three-form $\tilde{T}\in\Lambda^3\mathbb{K}^{N+1}$ must be 
        non-degenerate, i.e. it must define a non-degenerate second-order 
        homogeneous Hamiltonian operator;
    \item the alternating two-form $\tilde{A}\in\Lambda^2\mathbb{K}^{N+1}$ must be
        non-null, i.e. it must define a non-trivial system.
\end{enumerate}
Notice that point {(i)} implies that we can use the classification of second-order operators
obtained in~\cite{VerVit2}. That is, we can start with a fixed alternating
three-form in $\tilde{T}\in \Lambda^3\mathbb{K}^{N+1}$ in a standard form. So,
we can classify $\tilde{A}$ up to the action of subgroup of $\SL(N+1,\K)$
stabilising $\tilde{T}$, which we will denote by $\stab_{\SL(N+1,\K)}
(\tilde{T})$. 
}


{ Following the previous considerations, in this section we give a classification of systems of conservation laws possessing a second-order Hamiltonian formalism in
$N=2,4$ components and briefly discuss  the main conceptual and technical difficulties in the case $N=6$.} With $N=2,4$,  our results are valid for $\K$ being either
the real or the complex { field.} 

\subsection{Case $N=2$}

Let us start with the simplest case, i.e. $N=2$. The only operator of this 
type is  
\begin{equation}
    \label{pp2}
    \mathcal{P}_2
    =
    \begin{pmatrix}
        0&1\\-1&0
    \end{pmatrix}\partial_x^2,
\end{equation}
see~\cite{VerVit2}. The associated three-form is $\tilde{T}_2=\ud u^1\wedge
\ud u^2\wedge \ud u^3$. Note that $\stab_{\SL(3,\K)}(\tilde{T}_2)=\SL(3,\K)$.
Indeed, $\tilde{T}_2$ is a volume form of $\K^3$. In this case the
form $\Omega$, as in~\eqref{dec1}, is explicitly given by:
\begin{equation}\label{decN2}
    \Omega_2=\tilde{T}_2+A_{12}\ud u^1\wedge \ud u^2 \wedge \ud u^{4}
    +B_1 \ud u^1 \wedge \ud u^{3}\wedge \ud u^{4}+
    B_2 \ud u^2 \wedge \ud u^{3}\wedge \ud u^{4}.
\end{equation}
Note that $A_{12}\neq0$, otherwise we obtain only the null-system.
Therefore, using the action of $\SL(3,\K)$, we can rescale $A_{12}$ to $1$,
and map the vector $(B_1,B_2)^T$ to the vector $(1,0)^T$.
This exhausts our possibilities.

So, the { Hamiltonian system is}:
\begin{equation}
    \begin{cases}u^1_t=u^1_x,
    \\
    u^2_t=u^2_x,
    \end{cases}
\end{equation}
{i.e. it is linear.}

\subsection{Case $N=4$}

Let us firstly observe that there are two orbits of $\SL(5,\K)$ on three-forms, but one of them gives
$g_{ij}$ with determinant zero. We have to consider only the alternating
three-form in $\K^5$:
\begin{equation}
    \tilde{T}_4=\ud u^1\wedge \ud u^2\wedge \ud u^5+\ud u^3\wedge \ud u^4\wedge \ud u^5,
\end{equation}
which determines the constant operator, see~\cite{VerVit2}:
\begin{equation}
    \mathcal{P}_4=\begin{pmatrix}
        0&1&0&0\\
        -1&0&0&0\\
        0&0&0&1\\
        0&0&-1&0
    \end{pmatrix}\partial_x^2~.
    \label{eq:Pij4}
\end{equation}
Note that $\tilde{T}=\eta\wedge \ud u^5$, where 
\begin{equation}
    \eta=\ud u^1\wedge \ud u^2+\ud u^3\wedge \ud u^4
    \label{eq:etaform}
\end{equation}
is a symplectic form.  Therefore, as a set:
\begin{equation}
    \stab_{\SL(5,\K)}(\tilde{T}_{4})={\Sp(4,\K)}\rtimes \K^4.
    \label{eq:stabSL5}
\end{equation}
A matrix representation of this group is given
by the matrices $M\in\SL(5,\K)$ of the following form:
\begin{equation}
    M = \begin{pmatrix}
    C&0\\
    x^T&1
    \end{pmatrix},
    \quad C\in{\Sp(4,\K)},\,x\in\K^4,
\end{equation}
where we denoted by $0$ in the null vector in $\K^4$.

We first consider the action of {$\Sp(4,\K)$} on  $A\in\Lambda^2\K^4$.
The symplectic form $\eta$ on $\K^4$ introduces the following
splitting which is preserved by the action of {$\Sp(4,\K)$} on $\Lambda^2 \K^4$:
\begin{equation}
    \Lambda^2 \K^4 =  \K\eta\, \oplus\, \Theta_\eta,
    \label{eq:l2c4split}
\end{equation}
where:
\begin{equation}
    \Theta_\eta
    =\left\{\theta\in\Lambda^2\K^4 \,\middle|\, \eta\wedge\theta = 0 \right\}.
\end{equation}
To be explicit, an element $\theta\in\Theta_\eta$  can be uniquely written as:
\begin{equation}
  {  \begin{aligned}
    \theta &= \theta_0 \left(\ud u^1\wedge \ud u^2 -\ud u^3\wedge \ud u^4\right)
    +\theta_{1,3}\ud u^1\wedge \ud u^3
    \\
    &+\theta_{1,4}\ud u^1\wedge \ud u^4
    +\theta_{2,3}\ud u^2\wedge \ud u^3
    +\theta_{2,4}\ud u^2\wedge \ud u^4,
    \end{aligned}}
    \label{eq:thetasplit}
\end{equation}
and this gives a parametrization of the second factor in~\eqref{eq:l2c4split}.
Then, any element in the orbit of $\theta_\eta\eta+\theta$ is of the form
$\theta_\eta\eta+\theta'$ for some $\theta'\in \Theta_\eta$.

Moreover, there is a quadratic form $Q$ on the 5-dimensional subspace
$\Theta_\eta$, which is invariant under the action of {$\Sp(4,\K)$}, defined by
\begin{equation}
    \theta\wedge\theta\,=\,Q(\theta)\ud u^1\wedge \ud u^2\wedge \ud u^3\wedge \ud u^4,
    \label{eq:theta2}
\end{equation}
so, with $\theta$ as in equation~\eqref{eq:thetasplit} we have:
\begin{equation}
  { Q(\theta)\,=\,-2\theta_0^2 - 2\theta_{1,3}\theta_{2,4} + 2\theta_{1,4}\theta_{2,3}~.}
    \label{eq:Qform}
\end{equation}

\begin{remark}
    We remark that if we consider $\theta$ as an alternating $4\times 4$
    matrix, then $Q(\theta)$ is twice the Pfaffian of that matrix.
    \label{rem:pfaffian}
\end{remark}

It is well-known that the image of {$\Sp(4,\K)$} in the orthogonal group of $Q$ is
the connected component of the identity element and {$\Sp(4,\K)$} is the Spin double
cover of this component. Since the orthogonal group acts transitively on the
two-forms $\theta$ with a given value of $Q(\theta)$, the {$\Sp(4,\K)$}-orbit of
$\theta_\eta\eta+\theta$ consists of all the two-forms $\theta_\eta\eta+\theta'$
with $Q(\theta)=Q(\theta')$. Thus we may assume that:
\begin{equation}
 {   A\,=\,\theta_\eta \eta\,+\,\theta_{1,3}\ud u^1\wedge \ud u^3\,+\,\ud u^2\wedge \ud u^4}
    \label{eq:Acan}
\end{equation}
with {$Q(A)\,=\,-2\theta_{1,3}$}. That is, we fixed the two-form $A$, which now
depends on two arbitrary coefficients in $\K$, $\theta_{\eta}$ and {$\theta_{1,3}$}.

Now, we should use the remaining freedom to fix the shape of the vector $B$.
However, we note that this is superfluous, since being the cometric $g^{ij}$ in
equation~\eqref{eq:Pij4} constant following the definition of the vectors
$V^{i}$ in equation~\eqref{sysm} it will disappear upon differentiation. So,
using the definition of quasilinear system of conservation
laws~\eqref{11} we obtained that the only  system is
\begin{equation}
    \left\{  
        \begin{aligned}
            u_{t}^{1} &= \theta_{\eta}u^{1}_{x}-u^{4}_{x}, 
            \\
            u_{t}^{2} &=
            \theta_{\eta}u^{2}_{x}+\theta_{1, 3}u^{3}_{x},
            \\
            u_{t}^{3} &=u^{2}+ \theta_{\eta}u^{3},
            \\
            u_{t}^{4} &=-\theta_{1, 3}u^{1}_{x}+ \theta_{\eta}u^{4}_{x}.
        \end{aligned}
    \right.
    \label{eq:dim4sys}
\end{equation}
Note that the system is linear as in the $N=2$ case, but no longer decoupled.

{\subsection{Towards a classification for $N=6$}

In this subsection, we briefly discuss the main conceptual and technical
difficulties one encounters when addressing the problem of classifying quasilinear
systems of conservation laws admitting a homogeneous
second-order Hamiltonian formulation in 6 components. 

Following the
procedure discussed at the beginning of this section, we need first to fix
a second-order homogeneous Hamiltonian operator, to which is uniquely associated a three-form in dimension $7$.  The classification of three-forms in
dimension $7$ was undertaken by Schouten~\cite{Schouten}, and is for
instance reported in the book by Gurievich~\cite{gurievi}. Up to the action
of $\SL(7,\K)$, there are nine orbits, one which is open. To
classify the systems, we need to start from an orbit whose representative
is non-degenerate.  Following~\cite{VerVit2}, there are five such orbits,
including the open one. Therefore, to tackle this problem, we have to consider
a normal form for the three-form of shape:
\begin{equation}
\Omega = \omega + \omega' \wedge \ud u^8, \qquad \omega' := A_{ij} \ud
u^i \wedge \ud u^j \wedge + B_i \ud u^i \wedge \ud u^7,
\label{eq:form6}
\end{equation}
where $\omega$ is a \emph{properly chosen} representative from Schouten's list.
Here, we will just discuss briefly the first steps that are needed to solve
this problem in the case of the open orbit. We will show how many tools
from representation theory~\cite{Fulton_H} are needed and give a flavor of
how elaborate this process can be.

Let us consider the open orbit of the vector space $\Lambda^3\K^7$ with
respect to the action of $\SL(7,\K)$, and call its representative $\omega$
as in~\eqref{eq:form6}. The subgroup of $\SL(7,\K)$ that fixes $\omega$ is
a ($\K$-form of the) Lie group of type $G_2(\K)$ of dimension $14$, see
\cite[Prop.\ 22.12]{Fulton_H}.  Since $\omega$ is fixed by $G_2(\K)$, and
so is $\ud u^8$, we only need to understand the action of $G_2(\K)$ on
$\Lambda^2\K^7$.  This $G_2(\K)$-representation decomposes into irreducible 
representations of $G_2(\K)$ and of its Lie algebra $\mathfrak{g}_2$ as:
\begin{equation}
\Lambda^2\K^7\,=\,\mathfrak{g}_2\,\oplus\,\K^7~,    
\end{equation}
where as $G_2(\K)$-representations the two factors are the { adjoint and the 
standard }representation respectively, see \cite[\S 22.3]{Fulton_H}.

 At this point one has to understand the action of $G_2(\K)$ on the two factors in order to find invariants and normal forms.
These
computations necessitate quite some effort, so that  we will give a
detailed description of the whole procedure in an upcoming paper.
}

\section{Conclusions}
\label{concl}

{In this paper, we have established novel relationships among second-order
homogeneous Hamiltonian operators, their associated quasilinear systems
of conservation laws, and alternating three-forms along with their
intrinsic geometric structures. Specifically, we have proven the
following results:}
\begin{itemize}
    \item a set-theoretical bijection between $\mathcal{Y}_{N}$, the space of pairs
        $(\mathcal{P},V)$ with $\mathcal{P}$ second-order Hamiltonian operator
        and $V$ the associated system of conservation laws in $N$ components,
        and $\Lambda^{3}\K^{N+2}$, the space of alternating three-form in
        dimension $N+2$;
    \item the projective{-reciprocal} invariance of quasilinear system of conservation
        laws $V$ admitting second-order Hamiltonian structure, justifying the
        above bijection, and leading to a deeper geometric interpretation
        of the pairs $(\mathcal{P},V)$;
         \item a novel interpretation of the line congruence associated to a
        Hamiltonian quasilinear system $V$ as the annihilation set of lines
        of the alternating three-form corresponding to its pair;
    \item a classification of the {systems of conservation laws admitting a second-order Hamiltonian structure } for $N=2$ and
        $N=4$ through the action of the subgroup of $\SL(N+1,\K)$ stabilizing
        $\Lambda^{3}\K^{N+1}\subset \Lambda^{3}\K^{N+2}$ corresponding to the
        operator $\mathcal{P}$;
        \item {we showed that for $N=6$ such classification program requires 
            the introduction of tools from representation theory of Lie groups and algebras,
            by sketching the steps needed in the case of the open orbit.}
\end{itemize}

{Several} interesting open questions and possible {extensions} remain to be
investigated{, for instance:}
\begin{itemize}
    \item characterize the (classes of) {Hamiltonian systems for}
        $N>4$;
    \item extend the correspondence we found to the bi-Hamiltonian case, see
        for instance~\cite{Ferguson2008}, and the recent results
        of~\cite{LorVit1,LorVit2};
    \item extend the correspondence of the operator-system pairs to higher
        orders.
\end{itemize}


Finally, regarding the last suggested open problem, we observed that a similar
construction of congruences of lines has been presented
in~\cite{FPV17:_system_cl} for third-order Hamiltonian operators. However, a
direct interpretation as annihilation set of lines was not possible, because
the obtained linear equations are not governed by a totally skew-symmetric
tensor. The results of the present paper suggest that there might be a canonical
$(0,2)$-tensor which maps such linear equations into the annihilation set of
lines of a properly chosen alternating three-form. { We plan to investigate the general picture behind such operators in a future paper.}

\subsection*{Acknowledgements}


GG's and PV's research was partially supported by GNFM of the Istituto
Nazionale di Alta Matematica (INdAM), the research project Mathematical Methods
in Non-Linear Physics (MMNLP) by the Commissione Scientifica Nazionale – Gruppo
4 – Fisica Teorica of the Istituto Nazionale di Fisica Nucleare (INFN). 

PV's acknowledges the financial support of INdAM by the research grant
\emph{Borsa per l’estero dell’Istituto Nazionale di Alta Matematica}, { and is partially supported by the project “An artificial intelligence approach for risk assessment and prevention of low back pain: towards precision spine care”, PNRR-MAD-2022-12376692, CUP: J43C22001510001 funded by the European Union - Next Generation EU - NRRP M6C2 - Investment 2.1 Enhancement and strengthening of biomedical research in the NHS.}

\end{document}